\documentclass[a4paper]{article}

\usepackage[american]{babel}
\usepackage{amsmath,amssymb,amsthm}
\usepackage{centernot,braket}
\usepackage{cancel}
\usepackage{lipsum}
\usepackage{calc}
\usepackage{xcolor}
\usepackage{csquotes}
\usepackage{comment}
\usepackage{fancyhdr}
\usepackage{titling}
\usepackage[textwidth=4cm]{todonotes}
\usepackage{hyperref}
\usepackage{nicematrix}

\usepackage{tikz}
\usetikzlibrary{cd}
\def\temp{&} \catcode`&=\active \let&=\temp

\newcommand{\TAB}{\quad}

\newcommand{\CSTAR}{\ensuremath{C^*}}

\DeclareMathSymbol{\mlq}{\mathord}{operators}{'134}
\DeclareMathSymbol{\mrq}{\mathord}{operators}{'42}
\newcommand{\mathquote}[1]{\mlq #1 \mrq}

\newcommand{\BLANK}{{-}}

\newenvironment{TIKZCD}{\[\begin{tikzcd}}{\end{tikzcd}\]\ignorespacesafterend}
\newcommand{\MATRIX}[1]{\left(\hskip \arraycolsep\begin{matrix}#1\end{matrix}\hskip \arraycolsep\right)}
\newcommand{\SMATRIX}[1]{\left(\begin{smallmatrix}#1\end{smallmatrix}\right)}

\renewcommand{\subset}{\subseteq}

\renewcommand{\phi}{\varphi}
\renewcommand{\epsilon}{\varepsilon}

\makeatletter
\newcommand{\xRightarrow}[2][]{\ext@arrow 0359\Rightarrowfill@{#1}{#2}}
\makeatother

\newcommand{\INTEGERS}{\mathbb{Z}}

\newcommand{\COMPLEX}{\mathbb{C}}

\DeclareMathOperator{\mono}{mono}
\DeclareMathOperator{\epi}{epi}

\DeclareMathOperator{\trace}{tr}

\newcommand{\OTHER}{\mathrm{other}}

\newcommand{\inputs}{\mathrm{in}}
\newcommand{\outputs}{\mathrm{out}}
\newcommand{\win}{\mathrm{win}}

\newcommand{\winningprob}{\omega}
\newcommand{\losing}{\mathrm{losing}}

\newcommand{\game}{\mathrm{game}}
\newcommand{\bias}{\mathrm{bias}}
\newcommand{\parity}{\mathrm{par}}
\newcommand{\CHSH}{\mathrm{CHSH}}

\newcommand{\Alice}{\operatorname{Alice}}
\newcommand{\Bob}{\operatorname{Bob}}
\newcommand{\player}{\operatorname{player}}
\DeclareMathOperator{\magicsquare}{magic\ square}
\DeclareMathOperator{\magicpenta}{magic\ pentagram}
\newcommand{\twoplayer}{\operatorname{two~player}}

\newcommand{\words}{words}
\newcommand{\ordertwo}{order-$2$}
\newcommand{\spectrum}{\sigma}

\DeclareMathOperator{\supp}{supp}

\DeclareMathOperator{\ucp}{ucp}

\DeclareMathOperator{\proj}{proj}

\newcommand{\ILEQ}{\trianglelefteq}

\newcommand{\CENTER}{\mathcal{Z}}

\newcommand{\CONTINUOUS}{C}

\newcommand{\tensor}{\mathbin{\otimes}}

\usepackage{setspace}
\hypersetup{colorlinks,citecolor=blue}
\setuptodonotes{color=blue!10,bordercolor=blue!10}

\tikzcdset{row sep/normal=3.6em}
\tikzcdset{column sep/normal=5em}
\tikzset{help lines/.style={very thin, color=lightgray, dashed}}
\tikzset{axis lines/.style={very thin, color=lightgray}}

\makeatletter
\def\mathcolor#1#{\@mathcolor{#1}}
\def\@mathcolor#1#2#3{%
  \protect\leavevmode
  \begingroup
    \color#1{#2}#3%
  \endgroup}
\makeatother

\let\originalleft\left
\let\originalright\right
\renewcommand{\left}{\mathopen{}\mathclose\bgroup\originalleft}
\renewcommand{\right}{\aftergroup\egroup\originalright}

\newtheoremstyle{MYBREAK}
  {}          
  {}          
  {\itshape}  
  {}          
  {\bfseries} 
  {:}         
  {\newline}  
  {}          
\newtheoremstyle{MYPLAIN}
  {\topsep}   
  {\topsep}   
  {\itshape}  
  {15pt}          
  {\bfseries} 
  {}         
  {5pt plus 1pt minus 1pt} 
  {}          

\theoremstyle{definition}
\newtheorem{DEF}{Definition}[section]

\theoremstyle{plain}
\newtheorem{PROP}[DEF]{Proposition}

\newtheorem{THM}[DEF]{Theorem}

\newtheorem{REM}[DEF]{Remark}

\theoremstyle{MYBREAK}

\title{The quantum commuting model (Ia):\\The CHSH game and other examples:\\Uniqueness of optimal states}
\author{Alexander Frei}
\date{\today}

\begin{document}
\maketitle

\begin{abstract}
  We present in this paper that the CHSH game admits one and only one optimal state and so remove all ambiguity of representations.
  More precisely, we use the well-known universal description of quantum commuting correlations as state space on the universal algebra for two player games,
  and so allows us to unambigiously compare quantum strategies as states an this common algebra. As such we find that the CHSH game leaves a single optimal state on this commong algebra.

  In turn passing to any non-minimal Stinespring dilation for this unique optimal state is the only source of ambiguity (including self-testing):\linebreak
  More precisely, any state on some operator algebra may be uniquely broken up into its minimal Stinespring dilation as an honest representation for the operator algebra followed by its vector state.
  Any other Stinespring dilation however arises simply as an extension of the minimal Stinespring dilation (i.e.~as an embedding of the minimal Hilbert space into some random ambient one).
  As such this manifests the only source of ambiguity appearing in most (but not all!) traditional self-testing results such as for the CHSH game as well as in plenty of similar examples.


  We then further demonstrate the simplicity of our arguments on the Mermin--Peres magic square and magic pentagram game.

  Most importantly however, we present this article as an illustration of operator algebraic techniques on optimal states and their quotients,\linebreak
  and we further pick up the results of the current article in another following one (currently under preparation) to derive a first robust \mbox{self-testing} result in the quantum commuting model.
\end{abstract}

\section{Stinespring dilation}
\label{STINESPRING}
The Stinespring dilation for states on operator algebras will play the major role of the current article ---
more precisely it provides a convenient approach for verifying uniqueness of optimal states.
As such we begin with an introduction for convenience of the reader new to operator algebras.

Meanwhile we moreover explain the issue of randomly chosen non-minimal Stinespring dilations (in random ambient spaces) which do not reflect any of the properties of the original state. As such this accounts for all of the ambiguity arising in most (but not all!) self-testing results in literature.

Consider a state on a unital operator algebra
\[
  \phi:A\to \COMPLEX:\TAB\phi\geq0,\TAB\phi(1)=1.
\]
From this we may define an inner product pairing
\[
  A\times A\to \COMPLEX:\TAB \braket{x|y}:= \phi(x^*y)
\]
which defines a Hilbert space. We formally denote the Hilbert space in braket notation such that the inner product conveniently reads as composition,
\[
  H=\mathquote{A\ket\phi}:\TAB  \phi(x^*y) = \bra\phi x^*y \ket\phi.
\]
This allows us to later split such expressions to our desire.%
\footnote{This can be made rigorous in the setting of Hilbert modules.}
Based on the Hilbert space we obtain an induced representation by left multiplication
\[
  A\to B\left(\overline{A\ket\phi}\right): \TAB a\big(x\ket\phi\big) := ax\ket\phi
\]
together with an induced vector state by $\ket\phi=1\ket\phi$:
\[
  B\left(\overline{A\ket\phi}\right)\to\COMPLEX:\TAB T\mapsto \bra\phi T\ket\phi.
\]
Both combined however retrieve the original state $\bra\phi a\ket\phi = \phi(a)$.\\
As such our state allows for a dilation to the vector state:
\setlength{\skip\footins}{1cm}
\begin{gather*}
  \begin{tikzcd}
    A \arrow[rr, "\phi(\BLANK)"]\drar &[-0.5cm]&[-0.5cm] \COMPLEX: \\
    & B\left(~\overline{A\ket\phi}=\mathrm{minimal~space}~\right)
    \urar[swap]{\bra\phi\BLANK\ket\phi}
  \end{tikzcd}
\end{gather*}
This defines the (so-called) minimal Stinespring dilation in the following sense:
Each dilation arises simply as an embedding of this minimal dilation into some random ambient space. More precisely, consider another dilation say
\[
\begin{tikzcd}
  A \arrow[rr, "\phi(\BLANK)"]\drar &[-0.5cm]&[-0.5cm] \COMPLEX. \\
  & B\left(~H=\mathrm{some~space}~\right) \urar[swap]{\bra{\phi'}\BLANK\ket{\phi'}}
\end{tikzcd}
\]
Then the space defines an ambient space for the minimal one above:\\
More precisely, the embedding is formally given by the isometry
\[
  A\ket\phi\subset A\ket{\phi'}\subset H:\TAB x\ket\phi= x\ket{\phi'},\TAB \ket\phi=\ket{\phi'}.
\]
Indeed one easily verifies using the above dilation
\[
  \bra\phi x^*y\ket\phi=\phi(x^*y)=\bra{\phi'}x^*y\ket{\phi'}.
\]
Using the orthogonal decomposition we may invoke matrix formalism
\begin{gather*}
  B\Big(H=\overline{A\ket\phi}\oplus\OTHER\Big)=
  \MATRIX{ B\left(\overline{A\ket\phi}\right) & B\left(\OTHER\to\overline{A\ket\phi}\right) \\ B\left(\overline{A\ket\phi}\to\OTHER\right) & B\Big(\OTHER\Big) }
\end{gather*}
and so the representation into the random ambient dilation reads:
\begin{gather*}
  A\to B\Big(A\ket\phi\oplus\OTHER\Big):\TAB a\mapsto \SMATRIX{a&*\\{*}&*},\\[2\jot]
  B\Big(A\ket\phi\oplus\OTHER\Big)\to\COMPLEX:\TAB \SMATRIX{a&*\\{*}&*}\mapsto\SMATRIX{\bra\phi&0}\SMATRIX{a&*\\{*}&*}\SMATRIX{\ket\phi\\0}.
\end{gather*}
As such we have found that each dilation simply arises as an embedding into some random ambient space, and the description in matrix formalism illustrates how these ambient spaces do not reflect any properties of the original state.

On the other hand, the Stinespring dilation defines also the minimal quotient to which the original state descends. More precisely, the minimal quotient is given by the kernel (see \cite[subsection 2.1]{KWASNIEWSKI-EXEL-CROSSED} and \cite[subsection 3.1]{KWASNIEWSKI-MEYER-ESSENTIAL-CROSSED})
\begin{gather*}
  \begin{tikzcd}
    0\rar & I(\phi)\rar & A\rar & B\left(\overline{A\ket\phi}\right):
  \end{tikzcd}\\
  \begin{tikzcd}
    A\rar & A/I(\phi)\subset B\left(\overline{A\ket\phi}\right)\rar{\phi(\BLANK)} & \COMPLEX
  \end{tikzcd}
\end{gather*}
while any other quotient lies in between:
\begin{gather*}
  I\ILEQ A:\TAB I\subset\ker(\phi)\implies I\subset I(\phi):\\[2\jot]
  \begin{tikzcd}
    A\rar & A/I\rar &
    A/I(\phi)
    \rar{\phi(\BLANK)} & \COMPLEX.
  \end{tikzcd}
\end{gather*}
Summarizing, every state decomposes uniquely as a representation onto its minimal quotient followed by its dilation as vector state
(a good mnemonic here is the analogous decomposition of morphisms into epi plus mono):
\begin{TIKZCD}
  A\rar{\mathrm{quotient}} & A/I(\phi)
  \rar{\mathquote{\bra\phi\BLANK\ket\phi}} & \COMPLEX:\TAB \phi=\mathquote{\mono\circ\epi}
\end{TIKZCD}
In turn, one may use the unique decomposition for comparison of states as we will illustrate in this article for the CHSH and the Mermin--Peres games.

Finally let us make another remark regarding the canonical decomposition:
Suppose one chooses some factorization via some unital completely positive map (in short ucp-map) instead of an honest representation (which accounts for a choice of POVM instead of some honest PVM):
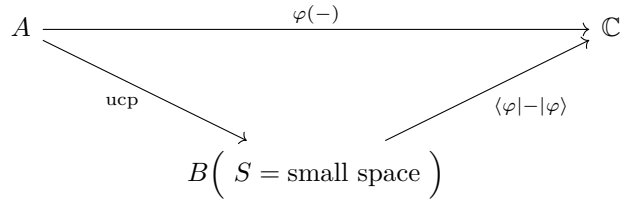
\begin{TIKZCD}
  A\arrow[rr,"\phi(\BLANK)"]\drar[swap]{\ucp} && \COMPLEX \\
  & B\Big(~S=\mathrm{small~space}~\Big) \urar[swap]{\bra\phi\BLANK\ket\phi}
\end{TIKZCD}
Using a minimal Stinespring dilation for the ucp-map we obtain an embedding into another ambient space
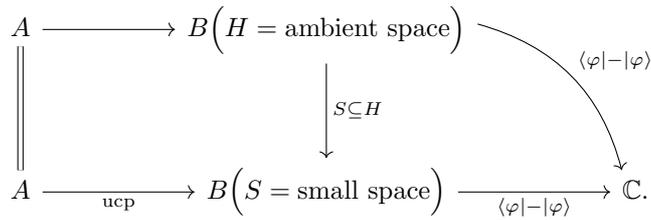
\begin{TIKZCD}
  A\rar\dar[equal] & B\Big(H=\mathrm{ambient~space}\Big)\dar{S\subset H}\drar[bend left]{\bra\phi\BLANK\ket\phi} \\
  A\rar[swap]{\ucp} & B\Big(S=\mathrm{small~space}\Big) \rar[swap]{\bra\phi\BLANK\ket\phi} & \COMPLEX.
\end{TIKZCD}
In other words, this may be understood as a mistaken decomposition:
\[
  \phi = (\mono\circ\mono)\circ \epi  = \mono \circ (\mono\circ\epi) = \bra\phi\BLANK\ket\phi\circ \ucp.
\]
Moreover, our previous discussion reveals this also as an ambient space for the original state, and as such reveals our \enquote{small space} as an incompatible cut:
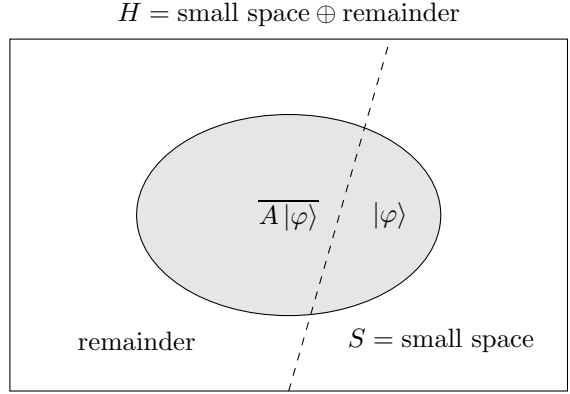
\begin{gather*}
  \begin{tikzpicture}[scale=2/3]
    \draw (-5.5,-3.5) rectangle (5.5,3.5);
    \draw[fill=black!10] (0,0) ellipse [x radius=3, y radius=2,rotate=0];
    \draw[dashed] (0,-3.5) -- (2,3.5);
    \node at (0,4) {$H=\mathrm{small~space}\oplus\mathrm{remainder}$};
    \node at (0,0) {$\overline{A\ket\phi}$};
    \node at (2,0) {$\ket\phi$};
    \node at (3,-2.5) {$S=\mathrm{small~space}$};
    \node at (-3,-2.5) {$\mathrm{remainder}$};
  \end{tikzpicture}
\end{gather*}
That is any choice of POVM (instead of an honest PVM) simply arises as the minimal Stinespring dilation embedded in some ambient space and then \textbf{cut incompatibly onto some small portion.}
As such this accounts for an additional, much worse source of ambiguity for self-testing.

This concludes our introduction and discussion on the Stinespring dilation from an operator algebraic perspective.
With this at hand, we may now proceed to nonlocal games and their optimal states --- and their uniqueness.

\section[Two-player algebra]{The two-player algebra}
\label{TWOPLAYER}

We continue with an introduction on the two-player algebra, which provides the universal setting for describing quantum correlations all as states on this common algebra --- and so allows us to compare them on a common ground.

It is well-known that the universal \CSTAR-algebra generated by a single projection valued measure agrees with the full group \CSTAR-algebra
\[
  \CSTAR(\INTEGERS/A)=\CSTAR\Big(e_1,\ldots,e_A\in\proj\Big|e_1+\ldots +e_A=1\Big)
\]
and as such we obtain as the universal \CSTAR-algebra per player for a given number of questions and answers the free product of groups
\[
  \CSTAR(X=\inputs|A=\outputs):=\CSTAR\Big(\INTEGERS/A*\ldots*\INTEGERS/A:\text{$X$-many}\Big).
\]
Finally a commuting pair of families of observables may be gathered algebraically by the maximal tensor product of their operator algebras and as such the algebra for two-player games has the universal description (with corresponding question and answer sets)
\[
  \CSTAR(\twoplayer):=\CSTAR(\Alice)\tensor\CSTAR(\Bob)=\CSTAR(X|A)\tensor\CSTAR(Y|B)
\]
and any quantum commuting strategy arises simply as a state
\[
  \phi:\CSTAR(\twoplayer)\to\COMPLEX:\TAB p(ab|xy)=\phi\Big(e(a|x)\tensor e(b|y)\Big).
\]
Indeed, applying any Stinespring dilation (see section \ref{STINESPRING}) one obtains the classical description of correlations as a commuting representation of projection valued measures together with a vector state,
\[
  \hspace{-1cm}\begin{tikzcd}[column sep = small]
    \CSTAR(\Alice)\tensor\CSTAR(\Bob) \arrow[rr, "\phi(\BLANK)\phantom{++++}"]\drar[swap]{\pi=\alpha\tensor\beta} &[1/2cm]&[1cm] \COMPLEX. \\
    & B\Big(~\overline{\CSTAR(\twoplayer)\ket\phi}\subset H~\Big) \urar[swap]{\bra\phi\BLANK\ket\phi}
  \end{tikzcd}
\]
This classical description however has the disadvantage that a Stinespring dilation introduces a level of ambiguity due to the chosen representation, which is basically the cause for the usual self-testing results.

Instead the description in terms of states on the two-player algebra allows now for a direct comparison of such on a common universal algebra, and so with this at hand we may now approach uniqueness of optimal states as follows.

\section[CHSH game]{CHSH game: single optimal state}\label{CHSH-GAME}
We proof in this section that the CHSH game has one and only one optimal state (compared to merely unique optimal correlation).
For this we begin with a quick review on the CHSH game phrased in the language of operator algebras.

Recall from \cite{CLAUSER-HORNE-SHIMONY-HOLT} that the CHSH game is a binary-input binary-output game, and as such the two-player algebra
is generated by a pair of \ordertwo\ unitaries for every player (compare section \ref{TWOPLAYER}):
\begin{gather*}
  \CSTAR(\twoplayer) = \CSTAR(\Alice)\tensor\CSTAR(\Bob):\\
  \CSTAR(\Alice)=\CSTAR(u^2=1=v^2)=\CSTAR(\Bob).
\end{gather*}
Note that the CHSH game is an XOR-game as introduced in \cite{CLEVE-HOYER-TONER-WATROUS}.\linebreak
As such we may equivalently describe the CHSH game by its bias polynomial
\[
  \win = \{a\oplus b =xy\}:\TAB \CHSH := u \tensor u + u \tensor v + v\tensor u - v\tensor v
\]
and an optimal state for the corresponding nonlocal game translates to maximizing the bias (see appendix \ref{XOR-GAMES} for a summary on XOR games):
\[
  \phi:\CSTAR(\twoplayer)\to\COMPLEX:\TAB \phi(\CHSH)=\|\CHSH\|.
\]
This concludes our quick review on the CHSH game. With this at hand, we may now proof the desired uniqueness of the optimal state for the CHSH game.

\begin{THM} The CHSH game has one and only one optimal state.\\
More precisely, consider the bias polynomial for the CHSH game
\[
  \CHSH=u\tensor u+ u\tensor v+v\tensor u -v\tensor v
\]
and suppose a state realises the maximal bias (see~proposition \ref{SYMMETRIC-SPECTRUM})
\[
  \phi:\CSTAR(\twoplayer)\to\COMPLEX:\TAB\phi(\CHSH)=\|\CHSH\|.
\]
Then the state necessarily factors over the matrix algebra
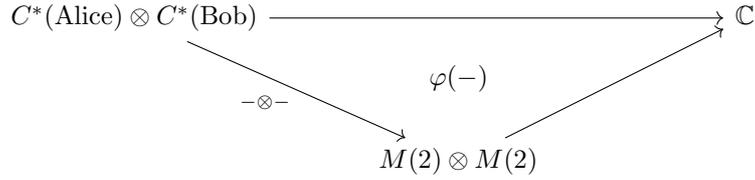
\begin{TIKZCD}
  \CSTAR(\Alice)\tensor\CSTAR(\Bob) \drar[swap]{\BLANK\tensor\BLANK} \arrow[rr, ""{name=phi}] &[-0.5cm]&[0.5cm] \COMPLEX \\
   & M(2)\tensor M(2) \urar\arrow[from=phi, phantom, "\phi(\BLANK)\phantom{++}"] &
\end{TIKZCD}
given by the canonical quotient (see proposition \ref{PAULI-ALGEBRA})
\begin{TIKZCD}[column sep=1cm]
  M(2)=\CSTAR(u^2=1=v^2|\{u,v\}=0).
\end{TIKZCD}
Note that the collection of words in generators reduces, by the anticommutation in the quotient, to the finite list of words (up to linear span)
\[
  \words(u,v)=\{1,u,v,uv=-vu,uvu=-v,\ldots\}\implies \{1,u,v,uv\}
\]
whence every state is already linearly determined on those:
\[
  \phi:M(2)\tensor M(2)\to\COMPLEX:\TAB \phi\Big(\{1,u,v,uv\}\tensor\{1,u,v,uv\}\Big)=\ ?
\]
With this observation in mind, the optimal state is uniquely and entirely determined (including all higher and mixed moments) by the correlation table
\[\def\arraystretch{1.3}\begin{array}{c|cccc}
    \phantom{+}\phi(\ldots)=\phantom{+} & \BLANK\tensor1 & \BLANK\tensor u & \BLANK\tensor v & \BLANK\tensor uv \\[\jot]
    \hline
    1\tensor\BLANK  & 1 & 0 &  0 & 0 \\
    u\tensor\BLANK  & 0      & \phantom{+}\tfrac{1}{\sqrt2}\phantom{+} &  \phantom{+}\tfrac{1}{\sqrt2}\phantom{+} & 0 \\
    v\tensor\BLANK  & 0      & \phantom{+}\tfrac{1}{\sqrt2}\phantom{+} & -\tfrac{1}{\sqrt2}\phantom{+} & 0 \\
    uv\tensor\BLANK & 0      & 0 &  0 & -1
  \end{array}\]
from which there is no ambiguity left anymore! (Recall section \ref{STINESPRING}.)\\
In short, there is one and only one optimal state for the CHSH game.
\end{THM}

Before we begin with the proof,
let us note how our notion of uniqueness on states compares with the traditional uniqueness on correlations:

\begin{REM}
  The entire correlation table for the optimal state above entails all quantum commuting correlations,
  since these are given by the state space restricted on two-moments of the form
  \[
    \begin{array}{ccccccc}
      \phi(1\tensor 1), & \phi(1\tensor u), & \phi(1\tensor v), &
      \cancel{\phi(1\tensor uv)}, & \cancel{\phi(1\tensor vu)}, & \cancel{\phi(1\tensor uvu)} & \cancel{\ldots} \\
      \phi(u\tensor 1), & \phi(u\tensor u), & \phi(u\tensor v), &
      \cancel{\phi(u\tensor uv)}, & \cancel{\phi(u\tensor vu)}, & \cancel{\phi(u\tensor uvu)} & \cancel{\ldots} \\
      \phi(v\tensor 1), & \phi(v\tensor u), & \phi(v\tensor v), &
      \cancel{\phi(v\tensor uv)}, & \cancel{\phi(v\tensor vu)}, & \cancel{\phi(v\tensor uvu)} & \cancel{\ldots} \\
      \cancel{\phi(uv\tensor 1)}, & \cancel{\phi(uv\tensor u)}, & \cancel{\phi(uv\tensor v)}, & \cancel{\ldots} \\
    \end{array}
  \]
  These values are equivalently determined by the traditional correlation table
  \[
    p\Big(a,b=0,1\Big|x,y=0,1\Big)=\phi\Big(E(a|x)\tensor E(b|y)\Big)
  \]
  given by the corresponding spectral projections
  \[
    u = E(0|0) - E(1|0),\TAB v = E(0|1) - E(1|1)
  \]
  which reduces to the familiar values for the CHSH game
  \begin{align*}
    p\Big(a\oplus b=xy\Big)     = \frac{1+1/\sqrt{2}}{4}, \TAB
    p\Big(a\oplus b\neq xy\Big) = \frac{1-1/\sqrt{2}}{4}.
  \end{align*}
  In particular, the traditional uniqueness of correlations --- given for example by self-testing --- excludes all higher and mixed moments of order-2 unitaries.
\end{REM}

\begin{proof}
Our approach follows the original one by Landau from \cite{LANDAU-CHSH-1} and its sequel \cite{LANDAU-CHSH-2}, which we however extend to uniqueness of states using some additional arguments from operator algebras.
We begin with the following observation: Suppose a state is optimal for the bias polynomial, then it is necessarily also optimal for the squared bias
\[
  \phi(\CHSH)=\|\CHSH\|\implies \phi\left(\CHSH^2\right)=\|\CHSH\|^2
\]
which indeed follows by Cauchy--Schwarz:%
\footnote{One may either use the bound $\ket\phi\bra\phi\leq \|\phi\|^2$ together with $a\leq b\implies c^*ac\leq c^*bc$,\linebreak
or equivalently the general Cauchy--Schwarz for ucp-maps $\phi(a)^*\phi(a)\leq\phi(a^*a)$ as originally proven by Choi in \cite{CHOI-1974}.
For a neat proof of the latter see also \cite[proposition 3.2]{PAULSEN-BOOK}.}
\begin{align*}
  \|\CHSH\|^2 &= \bra\phi\CHSH^*\ket\phi\bra\phi\CHSH\ket\phi \\
              &\leq\bra\phi\CHSH^*\CHSH\ket\phi\leq\|\CHSH\|^2
\end{align*}
In our case the squared bias however reduces to
\[
  \CHSH^2 = 4 - [u,v]\tensor [u,v]:\TAB \|\CHSH\|^2=8
\]
and so an optimal state necessarily satisfies
\[
  \phi\left(\CHSH^2\right)=8\implies \phi\Big([u,v]\tensor[u,v]\Big)=-4.
\]
We note that this defines only a necessary condition.\\
Arguing once more by Cauchy--Schwarz we obtain as above
\[
  \phi\Big([u,v]^*\tensor [u,v]^*[u,v]\tensor[u,v]\Big)=16
\]
and this once more is only a necessary but not sufficient condition.\\
Noting that $u^2=1=v^2$, this expression reads written out
\begin{align*}
  16 &= \phi\Big(( 2-uvuv-vuvu)\tensor (2-uvuv-vuvu)\Big) \\
     &= \phi(2\tensor 2) -\phi(2\tensor uvuv) - \ldots + \phi(vuvu\tensor vuvu).
\end{align*}
Note that each summand involves a unitary:
\[
  1\tensor1,\TAB 1\tensor uvuv,\TAB \ldots,\TAB  vuvu\tensor vuvu.
\]
Thus by a simple counting argument this is only possible if the state achieves the value plus/minus one with correct sign on each of the unitaries.
In particular
\[
  \phi(uvuv\tensor1)  = -1 = \phi(1\tensor uvuv).
\]
Within the Stinespring dilation we may however cut such relations into
\[
  \braket{x|y}:=\bra\phi uv\tensor1| uv\tensor1\ket\phi =\phi(uvuv\tensor1)=-1
\]
with both pieces of norm-one
\begin{align*}
  \braket{x|x}=\bra\phi uv\tensor1|vu\tensor1\ket\phi = 1\\
  \braket{y|y}=\bra\phi vu\tensor1|uv\tensor1\ket\phi=1
\end{align*}
and similarly for the second tensor factor $\phi(1\tensor uvuv)$.\\
That is however only possible if these are colinear with
\[
  uv\tensor1\ket\phi=- vu\tensor1\ket\phi,\TAB 1\tensor uv\ket\phi=- 1\tensor vu\ket\phi
\]
and so we have the anticommutation
\[
  \{u,v\}\tensor 1\ket\phi= 0=1\tensor \{u,v\}\ket\phi.
\]
Finally note that these anticommutators lie in the center
\[
  \{u,v\}\tensor1,1\tensor\{u,v\}\in\CENTER\CSTAR(\twoplayer)
\]
as one easily verifies on generators:
\[
  u\{u,v\} = v + uvu = \{u,v\}u,\TAB v\{u,v\} = vuv+v = \{u,v\}v.
\]
The anticommutator therefore also vanishes on the entire subspace
\begin{gather*}
  H := \overline{\CSTAR(\twoplayer)\ket\phi} = \overline{\CSTAR(\Alice)\tensor\CSTAR(\Bob)\ket\phi}:\\[2\jot]
  \{u,v\}\tensor1\ \CSTAR(\Alice)\tensor\CSTAR(\Bob)\ket\phi = 0\\
  1\tensor\{u,v\}\ \CSTAR(\Alice)\tensor\CSTAR(\Bob)\ket\phi = 0.
\end{gather*}
that is within the minimal Stinespring dilation (see section \ref{STINESPRING})
\[
  \CSTAR(\twoplayer)\to B(H=\min):\TAB \{u,v\}\tensor1=0=1\tensor \{u,v\}.
\]
In other words, the state factors as desired via (see proposition \ref{PAULI-ALGEBRA})
\[
  M(2)\tensor M(2) = \CSTAR\big(u^2=1=v^2\big|\{u,v\}=0\big)\tensor\CSTAR\big(u^2=1=v^2\big|\{u,v\}=0\big)
\]
which gives us the diagram for the optimal state:
\begin{TIKZCD}[column sep=small]
  \CSTAR(\Alice)\tensor\CSTAR(\Bob) \drar[swap]{\BLANK\tensor\BLANK} \arrow[rrr] &[-0.5cm]&[0.5cm]&[1cm] \COMPLEX. \\
   & M(2)\tensor M(2)\rar & B(H=\min)\urar[swap]{\bra\phi\BLANK\ket\phi}
\end{TIKZCD}
This finishes the first part of the proposition and so we may freely restrict our attention from now on to the optimal state already acting on the quotient
\begin{TIKZCD}
  \CSTAR(\Alice)\tensor\CSTAR(\Bob)\rar &[-4\jot] \mathrm{M}(2)\tensor\mathrm{M}(2) \rar{\phi(\BLANK)} &[1cm] \COMPLEX
\end{TIKZCD}
and so also the bias polynomial for the CHSH game
\[
  \CHSH= u\tensor(u+v) + v\tensor(u-v)\in\mathrm{M}(2)\tensor\mathrm{M}(2).
\]
From here the rest of the theorem basically follows from traditional results on the CHSH game.
We however present a new more algebraic approach below.\linebreak
But at first we untwist the CHSH polynomial. For this we note that the combined expressions (in brackets above) are easily seen to anticommute
\[
  \{u+v,u-v\} = 1 - \{u,v\}+\{u,v\} - 1=0.
\]
Due to the anticommutation of our generators,
these define also unitaries of order-2 (up to normalization)
\[
  (u \pm v)^2= u^2 \pm (uv + vu) + v^2 = 1 \pm \{u,v\} +1 = 2.
\]
In other words, these define just another abstract pair of Pauli matrices and so they give rise to the automorphism (see proposition \ref{PAULI-ALGEBRA})
\[
  M(2)\to M(2):\TAB u\mapsto \frac{u+v}{\sqrt2},\TAB v\mapsto\frac{u-v}{\sqrt2}
\]
which is at the same time its own inverse.
Using this automorphism we obtain the untwisted bias polynomial (basically in the new basis)
\[
  \CHSH/\sqrt2 = u\tensor u + v\tensor v.
\]
From here we give a new algebraic approach by verifying the trace property from which we deduce the final uniqueness,
and from which we may moreover easily read off the values for the correlation table:
Suppose our state maximises the untwisted bias polynomial and so also the unitary summands
\[
  \phi(u\tensor u+v\tensor v)=2\implies \phi(u\tensor u)=1=\phi(v\tensor v).
\]
From this it follows within the Stinespring dilation (similar as above):
\begin{align*}
  u\tensor 1\ket\phi = 1\tensor u\ket\phi,\TAB
  v\tensor 1\ket\phi = 1\tensor v\ket\phi.
\end{align*}
Indeed one easily verifies (and similarly for the other generator)
\[
  \bra\phi(u\tensor1-1\tensor u)^*(u\tensor1-1\tensor u)\ket\phi=\bra\phi2(1\tensor1-u\tensor u)\ket\phi=0.
\]
The entire trace property follows now by iteration on words of generators
\begin{align*}
  (u_1\ldots u_n u)\tensor1\ket\phi &= (u_1\ldots u_n)\tensor u\ket\phi= \ldots = 1\tensor (u u_n\ldots u_1)\ket\phi
\end{align*}
and we note that the order gets flipped (which is a cause for the transpose).\linebreak
Full matrix algebras however carry a unique tracial state and so we obtain as desired the unique optimal state for the CHSH game:
\[
  \phi\Big(u_1\ldots u_m\tensor v_1\ldots v_n\Big) = \tau\Big(u_1\ldots u_mv_n\ldots v_1\Big).
\]
The particular values for the correlation table may now be easily derived from this relation: Indeed we obtain for example
\[
  \phi(u\tensor v)=\tau\Big(u\cdot\frac{u+v}{\sqrt2}\Big)=\frac{1}{\sqrt2}\tau(1) + \frac{1}{\sqrt2}\tau(uv) = \frac{1}{\sqrt2}.
\]
This concludes the proof whence the optimal state is unique.
\end{proof}



\section[Mermin--Peres games]{Mermin--Peres games: single perfect state}
We continue in this section with another example of nonlocal games:\linebreak the Mermin--Peres magic square and magic pentagram games.
Both admit a unique winning state for rather obvious reasons (see the theorem below).\linebreak
In order to understand this uniqueness argument however, we will need a more thorough introduction first to linear constraint system games and their game algebra (and we note that the following arguments also apply to nonlinear constraint system games with any complex variables).

A linear constraint system game is a nonlocal game given by a system of linear equations (which we formulate in multiplicative form) such as
\begin{gather*}
  u,v,\ldots\in\exp\{2\pi i~\INTEGERS/2\}:\TAB uvw=1,\TAB vz=-1,\TAB\ldots
\end{gather*}
The players are asked an equation each and so
\[
  X=\left\{uvw=1,vz=-1,\ldots\right\}=Y
\]
and each player is asked to fill out the variables of their equation say
\[
  x=\left(vz=-1\right)\TAB\implies\TAB a=\left(v=-1,z=1\right).
\]
As a first rule the players are asked to answer in a synchronous way:
\[
  x=\mathrm{some~equation}=y\implies a=b.
\]
As such a winning state (being synchronous) defines a tracial state on the algebra for a single player (see \cite[corollary 5.6]{PSSTW-2016}).
More precisely, consider the two-player algebra (from section \ref{TWOPLAYER}) for identical algebras per player
\begin{gather*}
  \CSTAR(\twoplayer)=\CSTAR(\player)\tensor\CSTAR(\player)
\end{gather*}
and consider a winning state
\[
  \phi:\CSTAR(\player)\tensor\CSTAR(\player)\to\COMPLEX:\TAB\phi\Big(E(\losing)\Big)=0.
\]
Then the synchronicity rule enforces on the Stinespring dilation:
\[
  \phi\Big(E(a\neq b|x=y)\Big)=0\implies E(a|x)\tensor1\ket\phi =
  1\tensor E(a|x)\ket\phi.
\]
By iteration such states enjoy the special property (note the flip in order)
\begin{gather*}
  \phi\Big(E(a_1|x_1)\ldots E(a_m|x_m)\tensor E(b_1|y_1)\ldots E(b_n|y_n)\Big) \\
  = \phi\Big(E(a_1|x_1)\ldots E(a_m|x_m)E(b_n|y_n)\ldots E(b_1|y_1)\tensor1\Big)
\end{gather*}
and similar for flipping words in other ways from left to right.\\
Thus such states define a tracial state on either algebra per player
\begin{align*}
  \phi(\BLANK\tensor1)&:\CSTAR(\player)\tensor1\to\COMPLEX\\
  \phi(1\tensor\BLANK)&:1\tensor\CSTAR(\player)\to\COMPLEX
\end{align*}
which both allow us to recover the original state from above relation.\\
So we may from now on restrict our attention to a tracial state
\[
  \tau:\CSTAR(\player)\to\COMPLEX:\TAB\tau(ab)=\tau(ba).
\]
For tracial states however every additional rule enforces the corresponding \mbox{2-moments} of losing pairs to vanish in the minimal Stinespring dilation (see \cite{ORTIZ-PAULSEN-GRAPH} and \cite{HELTON-MEYER-PAULSEN-SATRIANO}). As such we are left to determine the game algebra for our synchronous game
\begin{gather}\label{GAME-ALGEBRA}
  \CSTAR(\mathrm{sync~game}) := \CSTAR\Big(\player\Big|E(\losing)=0\Big)
\end{gather}
where we denote for shorthand (similar as in appendix \ref{XOR-GAMES})
\[
  E(ab|xy):= E(a|x)E(b|y):\TAB E(S)=\{E(s)|s\in S\}
\]
and any tracial state on the quotient defines a winning correlation:
\begin{gather}\label{SYNC=TRACES}
  \tau:\CSTAR(\mathrm{sync~game})\to\COMPLEX:\TAB p(ab|xy)=\tau\Big(E(a|x)E(b|y)\Big).
\end{gather}
With this at hand we now get to the first particular rule for linear constraint system games.
For this we first note that for each equation as input we obtain a projection valued measure whence commuting projections for each variable. For the equation $x=(uvw=1)$ we obtain for example the triple
\[
  E(u\in\mathquote{\INTEGERS/2}|x),\TAB
  E(v\in\mathquote{\INTEGERS/2}|x),\TAB
  E(w\in\mathquote{\INTEGERS/2}|x)
\]
which we retrieve from the projection valued measure as a disjoint union like
\[
  E(v=-1|x):=\sum E\Big(u\in\mathquote{\INTEGERS/2},v=-1,w\in\mathquote{\INTEGERS/2}\Big|x\Big)
\]
and we recover the projection valued measure simply as intersection
\[
  E(u,v,w|x) = E(u|x) E(v|x) E(w|x).
\]
This is all easiest remembered as a Venn diagram (for commuting projections):
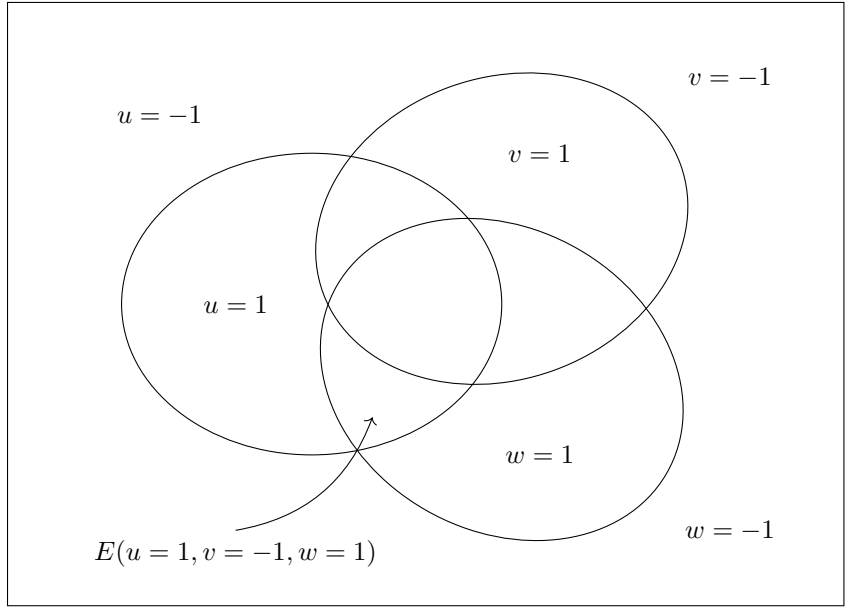
\begin{center}
  \begin{tikzpicture}
    \draw (-1,1) ellipse [x radius=2.5, y radius=2,rotate=0];
    \draw (3/2,2) ellipse [x radius=2.5, y radius=2,rotate=20];
    \draw (3/2,0) ellipse [x radius=2.5, y radius=2,rotate=-30];
    \draw (-5,-3) rectangle (6,5);
    \node at (-2,1) {$u=1$};
    \node at (-3,3.5) {$u=-1$};
    \node at (2,3) {$v=1$};
    \node at (4.5,4) {$v=-1$};
    \node at (2,-1) {$w=1$};
    \node at (4.5,-2) {$w=-1$};
    \node at (-2,-2.3) {$E(u=1,v=-1,w=1)$};
    \draw[->] (-2,-2) to[bend right=30] (-0.2,-0.5);
  \end{tikzpicture}
\end{center}
The first particular rule for linear constraint system games asks now the players to answer consistently for shared variables, for example
\[
  x=(uvw=1),\TAB y=(vz=-1)\TAB\implies\TAB\mathquote{a(v)=b(v)}
\]
which boils down to the equality for our projections (inside the game algebra).
\[
  E\Big(v\Big|x=(uvw=1)\Big) = E\Big(v\Big|y=(vz=-1)\Big).
\]
Indeed this follows very quickly from (think in terms of Venn-diagrams!)
\[
  E(v=1|x)E(v=-1|y) =0= E(v=-1|x)E(v=1|y).
\]
That is the projections are all independent of the chosen equation (cf.~\cite{CLEVE-MITTAL}).\\
With this at hand we may now get to the second particular rule for linear constraint system games.
For this we first pass to their generating unitaries
\[
  U:= E(u=1)-E(u=-1),\TAB V:= E(v=1)-E(v=-1),\TAB\ldots
\]
We now ask these unitaries to satisfy each equation, for example
\[
  f(u,v,w)= uvw = 1\TAB\implies\TAB f(U,V,W)=UVW=1.
\]
Note that the unitaries appearing in a common equation arose as commuting projections. As such they (isomorphically) generate the abelian subalgebra
\[
  \CONTINUOUS(\spectrum U)\tensor\CONTINUOUS(\spectrum V)\tensor\CONTINUOUS(\spectrum W) \subset\CSTAR(\player)
\]
and we may have instead also asked for any continuous function such as
\[
  g(U,V,W)=\cos(UV^3) + W^2-1=0.
\]
The generated subalgebra however agrees with the isomorphic copy
\[
  \CSTAR(g)=\CONTINUOUS(\supp g)E(\supp g)\subset\CSTAR(\player)
\]
and so modding out (the ideal generated) by either side results in the same quotient,
that is the requirement above reduces to its support projection:
\[
  g(U,V,W)=0\TAB\iff\TAB E(\supp g)=0.
\]
That is we obtain for our example (read backwards)
\[
  E(u,v,w)=E(u)E(v)E(w)=0,\TAB \forall uvw\neq1 \TAB\iff\TAB UVW=1.
\]
This is nothing but our second particular rule for linear constraint system games:
The players need to respond with variables satisfying their equation like
\[
  x = (uvw=1)\TAB\implies\TAB a=(u,v,w=\pm1):\TAB a(u)a(v)a(w)=1.
\]
Thus the game algebra \eqref{GAME-ALGEBRA} for linear constraint systems reads (see also \cite{CLEVE-MITTAL}):
\begin{align*}
  &\CSTAR\Big(\mathquote{uvw=1,vz=-1,\ldots}\Big)\\
  &=\CSTAR\Big(U^2=1,V^2=1,\ldots\Big| UVW=1: [U,V]=0,\ldots; VZ=-1:\ldots\Big).
\end{align*}
While allowing also for continuous functions, note that there was also nothing special about unitaries of \ordertwo. Instead all the above also holds perfectly fine for variables with values in any (possibly different) list of complex numbers
\[
  u\in\{\alpha_1,\ldots,\alpha_u\},\TAB v\in\{\beta_1,\ldots,\beta_v\},\TAB\ldots
\]
upon replacing the generators with
\[
  U=\alpha_1 E(u=\alpha_1) + \ldots + \alpha_u E(u=\alpha_u).
\]
This concludes the introduction to linear constraint system games.\\
With this at hand one may now very easily settle uniqueness of winning states:

\begin{THM}
    The game algebra \eqref{GAME-ALGEBRA} for the Mermin--Peres magic square and magic pentagram game read respectively
    \begin{gather*}
      \CSTAR(\magicsquare) = \mathrm{M}_2\tensor\mathrm{M}_2\\
      \CSTAR(\magicpenta) = \mathrm{M}_2\tensor\mathrm{M}_2\tensor\mathrm{M}_2
    \end{gather*}
    and as such both admit one and only one winning state given by the unique tracial state on either matrix algebra:
    \begin{gather*}
      T(\mathrm{M}_2\tensor\mathrm{M}_2)=\{\trace/2\tensor\trace/2\}\\
      T(\mathrm{M}_2\tensor\mathrm{M}_2\tensor\mathrm{M}_2)=\{\trace/2\tensor\trace/2\tensor\trace/2\}
    \end{gather*}
    As for the CHSH game, this entails all quantum commuting strategies.
\end{THM}
\begin{proof}
    The algebra for the Mermin--Peres magic square is generated by
    \begin{TIKZCD}
      u\tensor1\rar[thick,dash]\dar[thick,dash] &
      {\boxed{\text{North}}} \rar[thick,dash]\dar[thick,dashed, no head] &
      1\tensor u\dar[thick,dash] \\[2\jot]
      {\boxed{\text{West}}} \rar[thick,dash]\dar[thick,dash] &
      {\boxed{\text{Center}}} \rar[thick,dash]\dar[thick,dashed,no head] &
      {\boxed{\text{East}}} \dar[thick,dash] \\[2\jot]
      1\tensor v\rar[thick,dash] &
      {\boxed{\text{South}}} \rar[thick,dash] & v\tensor1
    \end{TIKZCD}
    where we have commuting rows and columns automatically built in as tensor product while using abstract copies of order two unitaries:
    \[
      u^2=1=v^2:\TAB u\tensor1,\TAB v\tensor1,\TAB 1\tensor u,\TAB 1\tensor v.
    \]
    Also all the boxes on outer outer edges are already determined:
    \[
      \boxed{\text{North}} = u\tensor u,\TAB \boxed{\text{West}} = u\tensor v,\TAB \boxed{\text{East}}= v\tensor u,\TAB \boxed{\text{South}} = v\tensor v.
    \]
    Using the center box one now easily derives their anticommutation:
    \begin{align*}
      \boxed{\text{West}}\cdot\boxed{\text{East}} = \boxed{\text{Center}} = -\boxed{\text{North}}\cdot\boxed{\text{South}} &\implies 1\tensor\{u,v\}=0,\\[2\jot]
      \boxed{\text{West}}\cdot\boxed{\text{East}} = \boxed{\text{Center}} = -\boxed{\text{South}}\cdot\boxed{\text{North}} &\implies \{u,v\}\tensor1=0.
    \end{align*}
    But the universal algebra generated by a pair of anticommuting \ordertwo\ unitaries is nothing but the abstract Pauli algebra from proposition \ref{PAULI-ALGEBRA} and so we found
    \[
      \CSTAR(\magicsquare) = \mathrm{M}_2\tensor\mathrm{M}_2.
    \]
    Similarly one argues for the magic pentagram and the proof is complete.
\end{proof}


\section*{Acknowledgements}

The author would like to thank his supervisor S{\o}ren Eilers for his kind support and encouragements,
as well as Connor Paddock for suggesting the quite valuable survey article \cite{SUPIC-BOWLES-SELF-TESTING}.
Moreover, the author acknowledges the support under the Marie–Curie Doctoral~Fellowship~No.~801199.

\appendix
\section{Pauli algebra}

\begin{PROP}\label{PAULI-ALGEBRA}
  The algebra of $2x2$-matrices (equivalently the algebra generated by its system of matrix units) has the equivalent description as the universal algebra generated by anticommuting \ordertwo\ unitaries
  \[
    M(2)=\CSTAR\Big(u^2=1=v^2\Big|\{u,v\}=0\Big)
  \]
  which identifies the generating order-2 unitaries as Pauli matrices
  \[
    u=\SMATRIX{1\\&-1}=Z,\TAB v=\SMATRIX{&1\\1}=X
  \]
  or any other (possibly rotated) pair of anticommuting Pauli matrices.\\
  We thus suggestively denote this as \textbf{the abstract algebra of Pauli matrices.}
\end{PROP}
\begin{proof}
  Note that the system of matrix units is equivalently generated by the partial isometry with the orthogonality relation for its range and source:
  \[
    w=\SMATRIX{&0\\1}:\TAB\SMATRIX{1\\&0}=\SMATRIX{&1\\0}\SMATRIX{&0\\1}\perp\SMATRIX{&0\\1}\SMATRIX{&0\\1}=\SMATRIX{0\\&1}.
  \]
  The algebra generated by such a partial isometry is now easily seen to be equivalent to the algebra generated by anticommuting \ordertwo\ unitaries.
  Indeed with the suggested identification from the proposition one has
  \[
    2w:=v(1+u)=\SMATRIX{&1\\1}\left(\SMATRIX{1\\&1}+\SMATRIX{1\\&-1}\right)=\SMATRIX{&0\\2}
  \]
  with inverse translation given by
  \begin{gather*}
    u:=w^*w-ww^*=\SMATRIX{&1\\0}\SMATRIX{&0\\1}-\SMATRIX{&0\\1}\SMATRIX{&1\\0}=\SMATRIX{1\\&-1}\\
    v:=w+w^*=\SMATRIX{&0\\1}+\SMATRIX{&1\\0}=\SMATRIX{&1\\1}
  \end{gather*}
  and one easily verifies the orthogonality relation for the range and source,\\
  as well as the anticommutation relation, respectively.
\end{proof}

\section[XOR-games]{Overview on XOR-games}
\label{XOR-GAMES}
We provide an overview of XOR-games for convenience of the reader, and we do so from an abstract, operator algebraic perspective.
Meanwhile, the author would like to acknowledge that the considerations of the current section arose as joint work with Azin Shahiri from a forthcoming work on the systematic (operator algebraic and representation-theoretic) classification of optimal states for the tilted CHSH games from \cite{LAWSON-LINDEN-POPESCU-TILTED-CHSH} and further \cite{ACIN-MASSAR-PIRONIO} with longterm goal for a better understanding of the I3322-inequality from \cite{FROISSART-I3322} and \cite{COLLINS-GISIN}.

An XOR-game is a nonlocal game in which the players answer with bits and such that the winning condition depends only on the parity of their answers:
\[
  A=\{0,1\}=B:\TAB\win=\{(a\oplus b|xy)|\ldots\}\subset \{0,1\}\times X\times Y.
\]
The most prominent example is of course the CHSH game as discussed in the current article.
To give another less obvious example however one may also consider the graph 2-coloring game, whose winning condition may be equivalently formulated as
\begin{gather*}
  x\sim y\implies a\oplus b = 1,\\
  x = y \implies a\oplus b = 0.
\end{gather*}
For convenience we consider bits from now on as $\pm1$-valued:
\[
  A=\{1,-1\}=B:\TAB a\oplus b=0,1\TAB\rightsquigarrow\TAB ab=\pm1.
\]
The point of such games is that they allow to verify the winning probability based on the bias instead, at which we take a closer look now.
Recall for this the two-player algebra with generators (see section \ref{TWOPLAYER}):
\begin{gather*}
  \CSTAR(\twoplayer)=\CSTAR(X|A=2)\tensor\CSTAR(Y|B=2):\\[\jot]
  u(x)= E(a=0|x)- E(a=1|x).
\end{gather*}
Now for nonlocal games in general the winning probability may be described by the operator norm for the game polynomial (which we take unnormalized)
\[
  \game := \sum E(\win):\TAB\winningprob_{qc}(\game)=\frac{\|\game\|}{|X\times Y|}
\]
and where we used the shorthand notation
\[
  E(ab|xy):= E(a|x)\tensor E(b|y):\TAB E(S)=\{E(s)|s\in S\}.
\]
For XOR-games we may now introduce the bias polynomial (also unnormalized)
\[
  \bias := \sum_{xy}\beta(x,y)u(x)\tensor u(y)
\]
with bias function given by the sum over the possible winning parities
\[
  \beta(x,y):=\sum\Big\{\parity=\pm1\Big|(\parity|xy)\in\win\Big\}=1,0,-1.
\]
For the 2-coloring example above this reads for example
\[
  \beta(x\sim y)=-1,\TAB\beta(x=y)=1,\TAB\beta(x\nsim y, x\neq y)=0.
\]
The point is now that the bias polynomial satisfies
\[
  2\game = |\win| + \bias.
\]
Recall that we take here the winning set as $\win\subset \{0,1\}\times X\times Y$.\\
Since the game polynomial is positive we obtain for the spectrum
\[
  0\leq\sigma(\game)\leq\|\game\|:\TAB \max\sigma(\game)=\|\game\|
\]
and by the relation above
\begin{gather*}
  2\max\spectrum(\game)=\max\spectrum(\bias) + |\win|.
\end{gather*}
We may therefore relate the winning probability to the maximal upper bound for the bias.
We remark that the above procedure --- that is passing from the norm to the spectrum --- may be seen as a linearization process since
\[
  \|a+1\|\neq\|a\|+1\TAB\rightsquigarrow\TAB\spectrum(a+1)=\spectrum(a)+1,\TAB \spectrum(\lambda a)=\lambda\spectrum(a).
\]
The goal is now to undo the linearization process --- that is passing from the spectrum for the bias polynomial back to its norm.
We achieve this with the following observation:

\begin{PROP}\label{SYMMETRIC-SPECTRUM}
  Bias polynomials for XOR-games have symmetric spectrum,
  from which the maximal bias agrees with the operator norm:
  \[
    \spectrum(\bias)=-\spectrum(\bias):\TAB \max\spectrum(\bias)=\|\bias\|.
  \]
  As a consequence for any XOR-game
  \[
    2\|\game\|=|\win|+\|\bias\|
  \]
  whence the winning probability relates to the problem (see \cite{FRITZ-NETZER-THOM}):\\
  \enquote{Can you compute the operator norm} for the bias polynomial?
\end{PROP}
\begin{proof}
  The automorphism which swaps the measurement outcomes for Alice and leaves the measurements for Bob untouched (or vice versa)
  \begin{gather*}
    \Phi:\CSTAR(X|2)\tensor\CSTAR(Y|2)\to\CSTAR(X|2)\tensor\CSTAR(Y|2):\\
    u(x)\tensor1\mapsto - u(x)\tensor1,
    \TAB
    1\tensor u(y)\mapsto 1\tensor u(y)
  \end{gather*}
  sends any bias polynomial to its negative
  \[
    \Phi(\bias) = -\sum \beta(xy)u(x)\tensor u(y) = -\bias.
  \]
  The spectrum however is invariant under automorphism whence
  \[
    \spectrum(\bias)=\spectrum(\Phi\bias)=\spectrum(-\bias)=-\spectrum(\bias).
  \]
  The remaining statements follow now from the symmetic spectrum.
\end{proof}
\begin{REM}
  One needs to stay cautious since polynomials (which do not strictly arise as bias polynomial) generally fail to have symmetric spectrum:\linebreak
  for example the I3322 inequality from \cite{FROISSART-I3322} and as reinvented in \cite{COLLINS-GISIN} might even have distinct values for the minimal and maximal bias.
\end{REM}

We finish this overview with an interesting negative result from \cite{CLEVE-HOYER-TONER-WATROUS},
and for which we note that it easily extends to include also quantum commuting strategies simply by using a Stinespring dilation as discussed in section \ref{STINESPRING}:

\begin{THM}[{\cite[theorem 3]{CLEVE-HOYER-TONER-WATROUS}}]
  Nonlocal games with binary outputs are never pseudo--telepathy games --- neither for quantum commuting strategies.
\end{THM}


\bibliography{Bibliography}
\bibliographystyle{alpha-all}

\end{document}